\documentclass[a4paper]{article}
\usepackage[T1]{fontenc}
\usepackage{graphicx}
\usepackage{color}
\usepackage[utf8]{inputenc}
\usepackage{amsmath,amsfonts,amsthm,mathtools,latexsym}
\usepackage{fullpage}
\usepackage{hyperref}
\usepackage{cite}

\usepackage{mdframed}

\newtheorem{theorem}{Theorem}
\newtheorem{proposition}{Proposition}
\newtheorem{lemma}{Lemma}

\title{Algorithmic Theory of Qubit Routing\thanks{This work was partially supported by 
JSPS KAKENHI Grant Numbers 
JP19K11814, 
JP20H05793, 
JP20H05795, 
JP20K11670, 
JP20K11692, 
JP22H05001, 
JP23K10982. 
}
\thanks{To appear in Proceedings of 18th Algorithms and Data Structures Symposium (WADS 2023).}%
}

\author{Takehiro Ito\thanks{Graduate School of Information Sciences, Tohoku University. E-mail: \texttt{takehiro@tohoku.ac.jp}}
    \and
    Naonori Kakimura\thanks{Faculty of Science and Technology, Keio University. E-mail: \texttt{kakimura@math.keio.ac.jp}} 
    \and
    Naoyuki Kamiyama\thanks{Institute of Mathematics for Industry, Kyushu University. E-mail: \texttt{kamiyama@imi.kyushu-u.ac.jp}}
    \and
    Yusuke Kobayashi\thanks{Research Institute for Mathematical Sciences, Kyoto University. E-mail: \texttt{yusuke@kurims.kyoto-u.ac.jp}}
    \and
    Yoshio Okamoto\thanks{Graduate School of Informatics and Engineering, The University of Electro-Communications. E-mail: \texttt{okamotoy@uec.ac.jp}}
    }
\date{}

\begin{document}

\maketitle

\begin{abstract}
The qubit routing problem, also known as the swap minimization problem, is a (classical) combinatorial optimization problem that arises in the design of compilers of quantum programs.
We study the qubit routing problem from the viewpoint of theoretical computer science, while most of the existing studies investigated the practical aspects.
We concentrate on the linear nearest neighbor (LNN) architectures of quantum computers, in which the graph topology is a path.
Our results are three-fold.
(1) We prove that the qubit routing problem is NP-hard.
(2) We give a fixed-parameter algorithm when the number of two-qubit gates is a parameter.
(3) We give a polynomial-time algorithm when each qubit is involved in at most one two-qubit gate.

\end{abstract}

\section{Introduction}

The qubit routing problem captures a (classical) combinatorial problem in designing compilers of quantum programs.
We rely on the formulation introduced by Siraichi et al.~\cite{DBLP:conf/cgo/SiraichiSCP18}.
In their setting, a quantum program is designed as a quantum circuit.
In a quantum circuit, there are wires and quantum gates such as Hadamard gates and controlled NOT gates. 
Each wire corresponds to one quantum bit (or a qubit for short) and gates operate on one or more qubits at a time.
A quantum circuit is designed at the logic level and needs to be implemented at the physical level.
This requires a mapping of logical qubits to physical qubits in such a way that all the gate operations can be performed.
However, due to physical restrictions, some sets of qubits could be mapped to places where an operation on those qubits cannot be physically performed.
This problem is essential for some of the superconducting quantum computers such as IBM Quantum systems.
Figure~\ref{fig:almaden1} shows the graph topology of such computers.

\begin{figure}[t]
    \centering
    \includegraphics[scale=0.75]{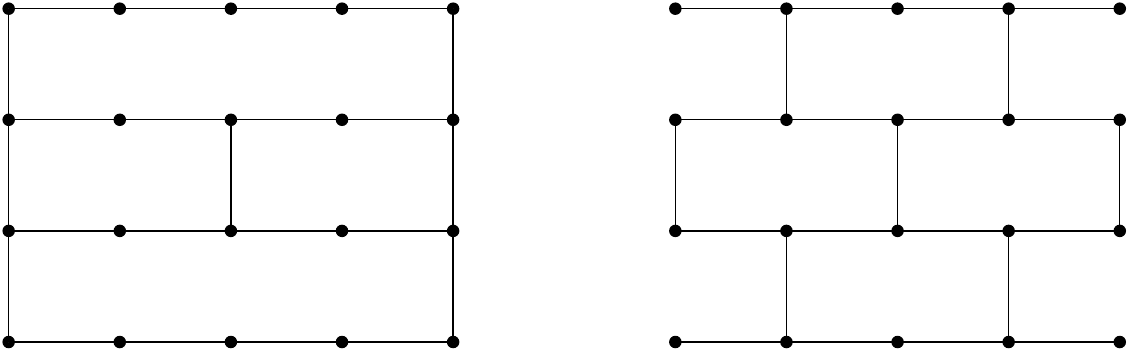}
    \caption{The graph topology of IBM quantum devices called ``Johannesburg'' (left) and ``Almaden'' (right). Each vertex represents a physical qubit, and we may perform a two-input gate operation only for a pair of adjacent qubits. In our problem formulation, the graph topology is taken into account as the graph $G$. Source: https://www.ibm.com/blogs/research/2019/09/quantum-computation-center/}
    \label{fig:almaden1}
\end{figure}

To overcome this issue, we insert so-called swap gates into the quantum circuit. A swap gate can swap two qubits on different wires, and can be implemented by a sequence of three controlled NOT gates.
With swap operations, we will be able to perform all the necessary gate operations that were designed in the original circuit.
The process is divided into two phases.
In the first phase, we find an initial mapping of logical qubits to physical qubits.
In the second phase, we determine how swap operations are performed.
Since a swap operation incurs a cost, we want to minimize the number of swap operations.
The quantum routing problem focuses on the second phase of the process.

Several methods have been proposed to solve the quantum routing problem.
The existing research has mainly focused on the practical aspects of the problem.
On the other hand, the theoretical investigation has been scarce.
Some authors have claimed that the problem is NP-hard, but
no formal proof was given to date.
Until this paper, it was not known whether the problem can be solved in polynomial time for the simplest case where the physical placement of qubits forms a path and the program has no two gate operations that can be performed at the same time.
A path corresponds to the case of the linear nearest neighbor architectures that have been extensively studied in the literature of quantum computing~(e.g.~\cite{DBLP:journals/qip/SaeediWD11}). 

This paper focuses on the theoretical aspects of the quantum routing problem and gives a better understanding of the problem from the viewpoint of theoretical computer science.
We are mainly concerned with the case where the physical placement of qubits forms a path.
Under this restriction, we obtain the following results.
\begin{enumerate}
    \item We prove that the qubit routing problem is NP-hard even when the program has no two gate operations that can be performed at the same time (Theorem~\ref{thm:hardnesspath}).
    \item We give a fixed-parameter algorithm when the number of gates in a given quantum circuit is a parameter (Theorem~\ref{thm:fpt}).
    \item We give a polynomial-time algorithm when each qubit is involved in at most one two-qubit operation (Theorem~\ref{thm:disjoint}). 
\end{enumerate}
As a side result, we also prove that the problem is NP-hard when the physical placement of qubits forms a star and any set of gate operations can be performed at the same time (Theorem~\ref{thm:hard-star}).

\paragraph{Problem Formulation}

We formulate the problem \textsc{Qubit Routing} in a purely combinatorial way as follows.
As input, we are given an undirected graph $G = (V, E)$, a partially ordered set (a poset for short) $P=(S, \preceq)$, a set $T$ of tokens, a map $\varphi\colon S \to \binom{T}{2}$, where $\binom{T}{2}$ denotes the set of unordered pairs of $T$, and an initial token placement $f_0\colon V \to T$, which is defined as a bijection.
The undirected graph $G$ corresponds to the physical architecture of a quantum computer in which each vertex corresponds to a physical qubit and an edge corresponds to a pair of qubits on which a gate operation can be performed.
The poset $P$ corresponds to a quantum circuit that we want to run on the quantum computer.
Each token in $T$ corresponds to a logical qubit.
The token placement $f_0$ corresponds to an initial mapping of the logical qubits in $T$ to the physical qubits in $V$ (e.g., as a result of the qubit allocation, see ``Related Work'' below).
For the notational convenience, we regard $f_0$ as a mapping from $V$ to $T$; this is not mathematically relevant since $f_0$ is bijective and to construct a mapping from $T$ to $V$, one uses the inverse $f_0^{-1}$. 
The bijectivity of a token placement is irrelevant since if there are more logic qubits than physical qubits, then we may introduce dummy logical qubits so that their numbers can be equal.
Each element in $S$ corresponds to the pair of logical qubits on which a gate operation is performed.
The correspondence is given by $\varphi$.
We note that $\varphi$ does not have to be injective.
The poset $P$ determines the order along which the gate operations determined by $\varphi$ are performed.
The order is partial since some pairs of operations may be performed independently: in that case, the corresponding elements of $S$ are incomparable in $P$.
An example is given in Figure~\ref{fig:circuit_example1}.

\begin{figure}[t]
    \centering
    \includegraphics{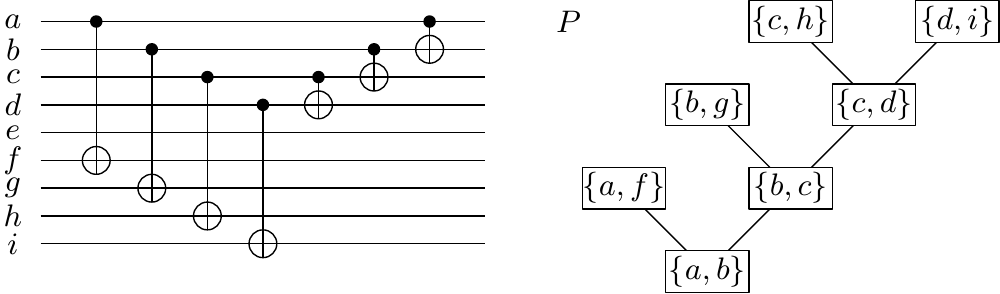}
    \caption{The left figure shows an instance of a quantum circuit, which is a part of the circuits given in \cite{DBLP:journals/qip/AsakaSY20}. The right figure shows the Hasse diagram of the corresponding poset $P$ in our problem formulation.}
    \label{fig:circuit_example1}
\end{figure}

A \emph{swap} $f_i \leadsto f_{i+1}$ is defined as an operation that transforms one token placement (i.e., a bijection) $f_i \colon V \to T$ to another token placement $f_{i+1}\colon V \to T$ such that there exists an edge $\{u,v\} \in E$ such that $f_{i+1}(u)=f_i(v)$, $f_{i+1}(v)=f_i(u)$ and $f_{i+1}(x) = f_i(x)$ for all $x \in V\setminus \{u,v\}$.
This corresponds to a swap operation of two logical qubits $f_i(u)$ and $f_i(v)$ assigned to two different physical qubits $u$ and $v$.

As an output of the qubit routing problem, we want a swap sequence $f_0 \leadsto f_1 \leadsto \dotsm \leadsto f_\ell$ that satisfies the following condition: 
there exists a map $i\colon S \to \{0,1,2,\dotsc,\ell\}$ such that $f_{i(s)}^{-1}(\varphi(s)) \in E$ for every $s \in S$ and if $s \preceq s'$, then $i(s) \leq i(s')$.
A swap sequence with this condition is said to be \emph{feasible}.  
The objective is to minimize the length $\ell$ of a swap sequence.
The condition states that all the operations in the quantum circuit that corresponds to the poset $P$ are performed as they follow the order $P$.

In summary, our problem is described as follows. 

\begin{mdframed}
\noindent
\textsc{Qubit Routing}
\begin{description}
\item[Input.] A graph $G=(V, E)$, a poset $P=(S, \preceq)$, 
a set $T$ of tokens, a map $\varphi\colon S \to \binom{T}{2}$, and an initial token placement $f_0\colon V \to T$. 
\item[Question.] 
Find a feasible swap sequence of minimum length. 
\end{description}
\end{mdframed}

Note that the minimum length of a feasible swap sequence is at most $|V||S|$ (if $G$ is connected), 
which implies that determining the existence of a feasible swap sequence of length bounded by a given number is in NP\@.

\paragraph{Related Work}

The qubit allocation problem was introduced by Siraichi et al.~\cite{DBLP:conf/cgo/SiraichiSCP18}.
Following Nannicini et al.~\cite{10.1145/3544563}, we divide the task in the qubit allocation problem into two phases.
The first phase involves qubit assignment that gives an initial placement of the logical qubits on the physical qubits, and the second phase involves qubit routing that inserts swap operations at appropriate positions so that all the gate operations in a given circuit can be performed.
The problem formulation in this paper solves the second-phase problem.
The qubit routing problem is often called the swap minimization, too.

For qubit routing, several heuristic algorithms have been given~\cite{DBLP:conf/cgo/SiraichiSCP18,DBLP:journals/pacmpl/SiraichiSCP19,DBLP:conf/asplos/LiDX19}, and integer-programming formulations have been given~\cite{10.1145/3544563,DBLP:journals/qip/HouteMACP20} that attempt to solve problem instances to optimality.
Siraichi et al.~\cite{DBLP:conf/cgo/SiraichiSCP18,DBLP:journals/pacmpl/SiraichiSCP19} have claimed that the qubit routing problem is NP-hard since it is more general than the so-called token swapping problem~\cite{DBLP:journals/tcs/YamanakaDIKKOSS15} that is known to be NP-hard~\cite{DBLP:conf/esa/MiltzowNORTU16,DBLP:journals/algorithmica/BonnetMR18,DBLP:journals/jgaa/KawaharaSY19} even for trees~\cite{DBLP:conf/esa/AichholzerDKLLM22}.
However, their argument was informal and no formal proof was not found.

A similar problem was studied by Botea, Kishimoto, and Marinescu \cite{DBLP:conf/socs/BoteaK018}.
In their problem, a mapping of logical qubits to physical qubits is injective but not bijective, and a logical qubit can only be moved to a physical qubit that is not occupied by another logical qubit.
They proved that the makespan minimization in their setting is NP-hard.


\section{Hardness: Paths and Chains}

In this section, we show that the problem is NP-hard even when $G$ is a path and $P$ is a chain (i.e., a totally ordered set). 

\begin{theorem}
\label{thm:hardnesspath}
 \textsc{Qubit Routing} is NP-hard even when $G$ is  a path and $P$ is a chain.
\end{theorem}

To prove \ref{thm:hardnesspath}, we first introduce notation.
For a swap sequence  $f_0 \leadsto f_1 \leadsto \dotsm \leadsto f_\ell$, which is denoted by $\mathbf{f}$,  
we say that $f_0$ and $f_\ell$ are the \emph{initial} and \emph{target} token placements of $\mathbf{f}$, respectively. 
We also say that $\mathbf{f}$ is from $f_0$ to $f_\ell$. 
The length $\ell$ of $\mathbf{f}$ is denoted by $|\mathbf{f}|$. 
For two swap sequences  $\mathbf{f}_1$ and $\mathbf{f}_2$, if the target token placement of $\mathbf{f}_1$ is equal to 
the initial token placement of $\mathbf{f}_2$, then its concatenation is denoted by $\mathbf{f}_1 \circ \mathbf{f}_2$.  
Note that in the concatenation of swap sequences, the target token placement of $\mathbf{f}_1$ and the initial token placement of $\mathbf{f}_2$ are identified, and thus $|\mathbf{f}_1 \circ \mathbf{f}_2| = |\mathbf{f}_1| + |\mathbf{f}_2|$.
When the initial and the target token placements of $\mathbf{f}$ coincide, for a positive integer $h$, 
we denote $\mathbf{f}^h  = \mathbf{f} \circ \mathbf{f} \circ \dots \circ \mathbf{f}$, where $\mathbf{f}$ appears $h$ times. 

Throughout this section, we only consider the case where $P$ is a chain. 
For a chain $P = (S, \preceq)$, let $s_1, s_2, \dots, s_{|S|}$ be distinct elements in $S$ such that $s_1 \prec s_2 \prec \dots \prec s_{|S|}$. 
Then, the information of $P$ and $\varphi$ can be represented as a sequence $Q = (q_1, q_2, \dots , q_{|S|})$, where $q_i := \varphi(s_i) \in \binom{T}{2}$ for each $i$. 
We say that a swap sequence  $f_0 \leadsto f_1 \leadsto \dotsm \leadsto f_\ell$ \emph{realizes} $Q$ if 
there exist $0 \le i_1 \le i_2 \le \dots  \le i_{|S|} \le \ell$ such that $f_{i_j}^{-1}(q_j) \in E$ for $j = 1, \dots , |S|$. 
In particular, if the swap sequence consisting of a single token placement $f$ realizes $Q$, then we say that $f$ realizes $Q$. 
With this terminology, when $P$ is a chain, \textsc{Qubit Routing} is to find a shortest swap sequence that realizes a given sequence of token pairs. 
For two sequences $Q_1$ and $Q_2$ of token pairs, its concatenation is denoted by $Q_1 \circ Q_2$. 
For a sequence $Q$ of token pairs and a positive integer $h$, we denote $Q^h = Q \circ Q \circ \dots \circ Q$, where $Q$ appears $h$ times.

To show the NP-hardness of \textsc{Qubit Routing}, we reduce \textsc{Optimal Linear Arrangement}, which is known to be NP-hard~\cite{GAREY1976237}. 

\begin{mdframed}
\noindent
\textsc{Optimal Linear Arrangement}
\begin{description}
\item[Input.] A graph $H=(V(H), E(H))$ and a positive integer $k$. 
\item[Question.] Is there a bijection $g\colon V(H) \to \{1, 2, \dots , |V(H)|\}$ that satisfies $\sum_{\{u,v\}\in E(H)} |g(u) - g(v)| \le k$?
\end{description}
\end{mdframed}

Suppose that we are given an instance of  \textsc{Optimal Linear Arrangement} that consists of a graph $H=(V(H), E(H))$ and a positive integer $k$. 
Denote $V(H) = \{v_1, v_2, \dots , v_n\}$ and $E(H) = \{e_1, e_2, \dots , e_m\}$, where $n = |V(H)|$ and $m=|E(H)|$. 
We may assume that $k < nm$ since otherwise, any bijection $g$ is a solution to \textsc{Optimal Linear Arragement}. 
Let $\alpha = 2nm+1$, which is an odd integer, and let 
$\beta$ and $\gamma$ be sufficiently large integers (e.g., $\beta = n^2 \alpha$ and $\gamma = 4k\alpha$).

We now construct an instance of \textsc{Qubit Routing} as follows.  
Define a set of tokens as $T=\{t_{v, i} \mid v \in V(H), \ i \in \{1, 2, \dots , \alpha\}\}$. 
Let $G = (V, E)$ be a path with $n \alpha$ vertices. 
We define
\begin{align*}
Q_v &:= (\{t_{v, 1}, t_{v, 2}\}, \{t_{v, 2}, t_{v, 3}\}, \dots , \{t_{v, \alpha-1}, t_{v, \alpha}\}) & & (v \in V(H)), \\
Q    &:= Q_{v_1} \circ Q_{v_2} \circ \dots  \circ Q_{v_n}, & & \\
\psi(e) &:= (\{ t_{u, nm+1}, t_{v, nm+1} \}) & & (e = \{u, v\} \in E(H)), \\
R &:= Q^\gamma \circ \psi(e_1) \circ Q^\gamma \circ \psi(e_2) \circ \dots \circ \psi(e_m) \circ Q^\gamma, & & 
\end{align*}
and $R^\beta$ is the sequence of token pairs that has to be realized. 
The initial token placement $f_0$ is defined arbitrarily. 
This gives an instance $(G, T,  R^\beta, f_0)$ of \textsc{Qubit Routing}.

To show the validity of the reduction, it suffices to show the following.
\begin{proposition}
\label{prop:reductioncorrectness}
The above instance $(G, T, R^\beta, f_0)$ of \textsc{Qubit Routing} has a solution of length less than $2 \beta (k \alpha + \alpha - m)$
if and only if the original instance of \textsc{Optimal Linear Arrangement} has a solution of objective value at most $k$. 
\end{proposition}

\subsection{Proof  of Proposition~\ref{prop:reductioncorrectness}}

To simpify the notation, we regard the vertex set $V$ of $G$ as $\{1, 2, \dots, n \alpha \}$ so that 
$\{i, i+1\} \in E$ for $i \in \{1, 2, \dots , n \alpha -1\}$. 
We say that a token placement $f \colon V \to T$ is \emph{block-aligned} if 
$|f^{-1}(t_{v, i}) - f^{-1}(t_{v, i+1})| =1$ for any $v \in V$ and for any $i \in \{1, 2, \dots , \alpha-1\}$. 
In other words, $f$ is block-aligned if and only if, for any $v \in V(H)$, the vertices $f^{-1}(t_{v, 1}), f^{-1}(t_{v, 2}), \dots , f^{-1}(t_{v, \alpha})$ appear consecutively in this order  (or in the reverse order) on the path $G$. 
Observe that $f$ is block-aligned if and only if it realizes $Q$. 
In what follows, we show the sufficiency and the necessity in Proposition~\ref{prop:reductioncorrectness} separately.

\subsubsection*{Sufficiency (``if'' part)}

Suppose that there exists 
a bijection $g\colon V(H) \to \{1, 2, \dots , |V(H)|\}$ such that $\sum_{\{u,v\}\in E(H)} |g(u) - g(v)| \le k$. 
Define a token placement 
$f^* \colon V \to T$ as $f^*((g(v) - 1) \alpha + i) = t_{v, {i}}$ for $v \in V(H)$ and $i \in \{1, 2, \dots , \alpha\}$. 
Then, $f^*$ is block-aligned, and hence $f^*$ realizes $Q$. 
This implies that $f^*$ realizes $Q^\gamma$, too. 

We use the following two lemmas. 

\begin{lemma}
\label{clm:11}
There exists a swap sequence $\mathbf{f}_0$ from $f_0$ to $f^*$ whose length is less than $(n \alpha)^2$. 
\end{lemma}

\begin{proof}
Since $|V| = |T| = n \alpha$, by applying swaps at most $n \alpha - 1$ times to $f_0$,
we obtain a token placement $f_1$ such that an end vertex $v$ of the path $G$ satisfies $f_1(v) = f^*(v)$. 
By applying the same operation for each vertex $v$ in $G$ from one end to the other, we obtain $f^*$. 
The total number of token swaps is at most $(n\alpha - 1) |V| < (n \alpha)^2$. 
\end{proof}

\begin{lemma}
\label{clm:12}
For any $e = \{u, v\} \in E(H)$, there exists a swap sequence  $\mathbf{f}_e$ from $f^*$ to $f^*$ such that 
$\mathbf{f}_e$ realizes $\psi(e) = (\{ t_{u, nm+1}, t_{v, nm+1} \})$ and $|\mathbf{f}_e| = 2 | g(u) - g(v)| \alpha - 2$. 
\end{lemma}

\begin{proof} 
By applying swaps $|(f^*)^{-1}(t_{u, nm+1}) - (f^*)^{-1}( t_{v, nm+1})| - 1 = |g(u) - g(v)| \alpha - 1$ times to $f^*$, 
we can obtain a token placement $f_e$ such that $(f_e)^{-1}(t_{u, nm+1})$ and $(f_e)^{-1}(t_{v, nm+1})$ are adjacent, i.e., $f_e$ realizes $\psi(e)$. 
Conversely, $f_e$ can be transformed to $f^*$ by applying $|g(u) - g(v)| \alpha - 1$ swaps. 
Therefore, there exists a swap sequence of length $2 | g(u) - g(v)| \alpha - 2$ that contains $f^*, f_e, f^*$ in this order, which completes the proof. 
\end{proof}

We now show that the following swap sequence satisfies the conditions: 
\[
\mathbf{f} := \mathbf{f}_0 \circ ( \mathbf{f}_{e_1} \circ \mathbf{f}_{e_2} \circ \dots \circ \mathbf{f}_{e_m})^\beta, 
\]
where $\mathbf{f}_0$ and $\mathbf{f}_{e_i}$ are as in Lemmas~\ref{clm:11} and \ref{clm:12}, respectively. 
Since $\mathbf{f}_{e_1} \circ \mathbf{f}_{e_2} \circ \dots \circ \mathbf{f}_{e_m}$ realizes $R$, we see that $\mathbf{f}$ realizes $R^\beta$. 
Furthermore, we obtain
\begin{align*}
|\mathbf{f}| 
&= |\mathbf{f}_0| + \beta \sum_{i=1}^m | \mathbf{f}_{e_i}|  \\ 
&< (n \alpha)^2 + \beta \sum_{\{u, v\} \in E(H)} (2 | g(u) - g(v)| \alpha - 2) \\
&\le (n \alpha)^2  + 2 \beta ( k \alpha -  m)  \\ 
&< 2 \beta ( \alpha + k \alpha -  m ). 
\end{align*}
This shows that $\mathbf{f}$ is a desired swap sequence. 

\subsubsection*{Necessity (``only if'' part)}

To show the necessity, we first show a few properties of block-aligned token placements. 

\begin{lemma}\label{clm:01}
If a token placement $f \colon V \to T$ is block-aligned, then there exists a bijection $g \colon V(H) \to \{1, 2, \dots , n\}$ such that 
\begin{equation}
f^{-1}(t_{v, {nm+1}}) = (g(v) - 1) \alpha + nm + 1 \text{ for any } v \in V(H). \label{eq:01}
\end{equation} 
\end{lemma}

\begin{proof}
Let $f \colon V \to T$ be a block-aligned token placement. 
Since $f^{-1}(t_{v, 1})$, $f^{-1}(t_{v, 2}), \dots , f^{-1}(t_{v, \alpha})$ appear consecutively in $G$ for every $v \in V(H)$, 
there exists a bijection $g \colon V(H) \to \{1, 2, \dots , n\}$ such that
\[
\{f^{-1}(t_{v, 1}), f^{-1}(t_{v, 2}), \dots , f^{-1}(t_{v, \alpha})\} = \{ (g(v) - 1) \alpha + 1, (g(v) - 1) \alpha + 2, \dots , g(v) \alpha \}
\]
for $v \in V(H)$. 
Since $f^{-1}(t_{v, 1}), f^{-1}(t_{v, 2}), \dots , f^{-1}(t_{v, \alpha})$ appear in this order or in the reverse order, 
$f^{-1}(t_{v, {nm+1}})$ has to be located in the middle of them in either case, where we note that $nm+1 = (\alpha +1)/2$. 
Therefore, $f^{-1}(t_{v, {nm+1}}) = (g(v) - 1) \alpha + nm + 1$. 
\end{proof}

We say that 
a block-aligned token placement $f \colon V \to T$ \emph{corresponds to} a bijection $g \colon V(H) \to \{1, 2, \dots , n\}$ 
if  (\ref{eq:01}) holds.

\begin{lemma}
\label{clm:05}
Let $g_1, g_2 \colon V(H) \to \{1, 2, \dots , n\}$ be bijections with $g_1 \neq g_2$. 
Suppose that $f_1$ and $f_2$ are block-aligned token placements that correspond to $g_1$ and $g_2$, respectively. 
Then, any swap sequence from $f_1$ and $f_2$ is of length at least $\alpha^2$. 
\end{lemma}

\begin{proof}
Since $g_1 \neq g_2$, there exists a vertex $v \in V(H)$ with $g_1(v) > g_2(v)$. 
Since
\[\{f_1^{-1}(t_{v, 1}), f_1^{-1}(t_{v, 2}), \dots , f_1^{-1}(t_{v, \alpha})\} = \{ (g_1(v) - 1) \alpha + 1, (g_1(v) - 1) \alpha + 2, \dots , g_1(v) \alpha \}\]
and 
\[\{f_2^{-1}(t_{v, 1}), f_2^{-1}(t_{v, 2}), \dots , f_2^{-1}(t_{v, \alpha})\} = \{ (g_2(v) - 1) \alpha + 1, (g_2(v) - 1) \alpha + 2, \dots , g_2(v) \alpha \},\] 
the length of any swap sequence from $f_1$ and $f_2$ is at least 
\begin{align*}
\sum_{i =1}^\alpha (f_1^{-1}(t_{v, i}) - f_2^{-1}(t_{v, i}) ) =  \alpha^2 (g_1(v) - g_2(v)) \ge \alpha^2, 
\end{align*}
which completes the proof. 
\end{proof}

Suppose that there exists a swap sequence $\mathbf{f}$ of length less than $2 \beta (k \alpha + \alpha -  m )$ that realizes $R^\beta$.
Then, there exists a subsequence $\mathbf{f'}$ of $\mathbf{f}$ such that
$\mathbf{f'}$ realizes $R$ and $| \mathbf{f'} | \le |\mathbf{f}| / \beta < 2 (k \alpha + \alpha - m)$. 
Since $|\mathbf{f'}| < 2 (k+1) \alpha \le \gamma$, if a subsequence of  $\mathbf{f'}$ realizes $Q^\gamma$, then it contains a token placement that realizes $Q$. 
With this observation, we see that 
$\mathbf{f'}$ contains token placements $f^*_1, f_{e_1}, f^*_2, f_{e_2}, f^*_3, \dots , f_{e_m}, f^*_{m+1}$ in this order, where 
$f^*_i$ realizes $Q$ (i.e., it is block-aligned) and $f_{e_i}$ realizes $\psi(e_i)$ for each $i$.

For $i \in \{1, 2, \dots , m+1\}$, Lemma~\ref{clm:01} shows that  $f^*_i$ corresponds to some bijection $g_i \colon V(H) \to \{1, 2, \dots , n\}$. 
Furthermore, since $| \mathbf{f'} | < 2 (k+1) \alpha \le \alpha^2$, Lemma~\ref{clm:05} shows that
$g_1 = g_2 = \dots = g_{m+1}$. That is,  $f^*_1, f^*_2, \dots , f^*_{m+1}$ correspond to a common bijection $g \colon V(H) \to \{1, 2, \dots , |V(H)|\}$. 

We now show that $g$ is a desired bijection. 
For every $e_i = \{u, v\}$, 
any swap sequence from $f^*_i$ to $f_{e_i}$ has length at least $|(f^*_i)^{-1}(t_{u, nm+1}) - (f^*_i)^{-1}(t_{v, nm+1})| - 1 = \alpha |g(u) - g(v)| - 1$ as $f^*_i$ corresponds to $g$. 
Similarly, any swap sequence from $f_{e_i}$ to $f^*_{i+1}$ has length at least $|(f^*_{i+1})^{-1}(t_{u, nm+1}) - (f^*_{i+1})^{-1}(t_{v, nm+1})| - 1 = \alpha |g(u) - g(v)| - 1$. 
Therefore, we obtain 
\[
|\mathbf{f'}| \ge \sum_{\{u, v\} \in E(H)} 2 (\alpha |g(u) - g(v)| - 1). 
\]
This together with $| \mathbf{f'} | < 2 (k \alpha + \alpha- m)$ shows that $\sum_{\{u, v\} \in E(H)} |g(u) - g(v)| < k+1$. 
This implies that $\sum_{\{u, v\} \in E(H)} |g(u) - g(v)| \le k$ by integrality, and hence $g$ is a desired bijection. 
This completes the proof of Theorem~\ref{thm:hardnesspath}. 
\qed

\section{Algorithm Parameterized by the Number of Gates}\label{sec:fpt}

In this section, we assume that $G=(V, E)$ is a path.
Let $k=|S|$ be the size of a poset $P=(S, \preceq)$.
The purpose of this section is to design a fixed-parameter algorithm for the problem parameterized by $k$.
Since $G$ is a path, we suppose for simplicity that $V=\{1,2,\dots, n\}$ and $E=\{\{i, i+1\}\mid i=1,2, \dots, n-1\}$.

We first observe that we may assume that a poset forms a chain.
Indeed, suppose that $P$ is not a chain.
Then, if we have a fixed-parameter algorithm for a chain, we can apply the algorithm for all the linear extensions of $P$.
Since the number of the linear extensions is at most $k!$, it is a fixed-parameter algorithm for $P$.
Thus, we may assume that $S=\{s_1, s_2, \dots, s_k\}$ such that $s_1 \prec s_2 \prec \dots \prec s_k$. 
Let $\tilde{T}=\bigcup_{s\in S} \varphi (s)$ and let $\tilde{k}=|\tilde{T}|$.
Then, we have $\tilde{k}\leq 2k$.
We denote $[\tilde k]=\{1, 2, \dots , \tilde k\}$. 

Let $f$ be a token placement.
Let $\tilde{V}$ be the positions where tokens in $\tilde{T}$ are placed, i.e., $\tilde{V}=\{f^{-1}(t) \mid t\in \tilde{T}\}$.
We denote $\tilde{V}=\{v_1, v_2, \dots, v_{\tilde{k}}\}$ where $v_1< v_2 < \dots < v_{\tilde{k}}$.
Define a vector $x\in \mathbb{Z}^{\tilde{k}}$ so that $x_i = v_{i+1} - v_{i}$, where $v_0 = 0$, for every index $i=0, 1, \dots, \tilde{k}-1$.
We note that $\sum_{i=0}^{\tilde{k}-1} x_i = v_{\tilde{k}} \leq n$ holds.
We further define a bijection $\sigma \colon [\tilde k] \to \tilde T$ as $\sigma (i) = f(v_i)$.
Let $\Sigma$ denote the set of all bijections from $[\tilde k]$ to $\tilde T$. 
We call the pair $(x, \sigma)$ the \emph{signature} of the token placement $f$, which is denoted by $\mathsf{sig}(f)$.
The signature maintains the information only on the tokens in $\tilde{T}$, which suffices for finding a shortest feasible swap sequence since swapping two tokens not in $\tilde{T}$ is redundant in the swap sequence.

We first present a polynomial-time algorithm when $k$ is a fixed constant in Section~\ref{sec:algorithm_constk}, and then a fixed-parameter algorithm in Section~\ref{sec:algorithm_fpt}.

\subsection{Polynomial-time algorithm for a fixed constant \texorpdfstring{$k$}{k}}\label{sec:algorithm_constk}

Define $\mathcal{R} = \{x\in \mathbb{Z}^{\tilde{k}}\mid \sum x_j \leq n, x_j\geq 1~(j=0, 1,\dots, \tilde{k}-1)\}$ and $\mathcal{S} = \{(x, \sigma)\mid x\in \mathcal{R}, \sigma \in \Sigma\}$.
We see that, for any $(x, \sigma)\in \mathcal{S}$, there exists a token placement $f$ such that $\mathsf{sig}(f)=(x, \sigma)$, and such a placement $f$ can be found in polynomial time even when $k$ is not a constant.
It holds that $|\mathcal{S}|=O(\tilde{k}! n^{\tilde{k}})$.

For two signatures $(x^0,\sigma^0)$ and $(x^t,\sigma^t)$, a \emph{swap sequence from $(x^0,\sigma^0)$ to $(x^t,\sigma^t)$} means a swap sequence $f_0 \leadsto f_1 \leadsto \dotsm \leadsto f_\ell$ for some token placements $f_0$ and $f_\ell$ such that $\mathsf{sig}(f_0)=(x^0,\sigma^0)$ and  $\mathsf{sig}(f_\ell)=(x^t,\sigma^t)$.
Its length is defined to be $\ell$.

We first show that we can find a shortest swap sequence between two signatures in polynomial time when $k$ is a fixed constant.

\begin{lemma}\label{lem:shorestpathFor2placements}
For two signatures $(x^0,\sigma^0)$ and $(x^t,\sigma^t)$,
we can find a shortest swap sequence from $(x^0,\sigma^0)$ to $(x^t,\sigma^t)$ in $O(\mathrm{poly}(|\mathcal{S}|))$ time.
\end{lemma}
\begin{proof}
We construct a graph on the vertex set $\mathcal{S}$ such that $(x, \sigma)$ and $(x', \sigma')$ are adjacent if and only if there exist two token placements $f$, $f'$ such that $\mathsf{sig}(f)=(x, \sigma)$, $\mathsf{sig}(f')=(x', \sigma')$, and $f\leadsto f'$.
Then a path from $(x^0,\sigma^0)$ to $(x^t,\sigma^t)$
in the graph corresponds to a swap sequence from $(x^0,\sigma^0)$ to $(x^t,\sigma^t)$.
Therefore, we can find a shortest swap sequence by finding a shortest path from $(x^0,\sigma^0)$ to $(x^t,\sigma^t)$ in the graph.
Since the number of vertices of the constructed graph is $|\mathcal{S}|$, 
it can be done in $O(\mathrm{poly}(|\mathcal{S}|))$ time. 
\end{proof}

The above lemma allows us to design a polynomial-time algorithm for a fixed constant $k$.

\begin{theorem}
Let $P= (S, \preceq)$ be a chain of $k$ elements, where $k$ is a fixed constant.
For a token placement $f_0$,
we can find a shortest feasible swap sequence from $f_0$ in polynomial time.
\end{theorem}
\begin{proof}
Let $\mathcal{S}_i = \{(x, \sigma)\in \mathcal{S}\mid \exists f\text{\ s.t.~} \mathsf{sig}(f)=(x, \sigma) \text{\ and }f^{-1}(\varphi(s_i))\in E\}$, which is the set of signatures $(x, \sigma)$ that correspond to token placements in which $\varphi(s_i)$ are adjacent.
Note that 
if $(x, \sigma)\in \mathcal{S}_i$, then 
$\mathsf{sig}(f)=(x, \sigma)$ implies $f^{-1}(\varphi(s_i))\in E$ for any $f$, 
because $\varphi(s_i) \subseteq \tilde T$. 
Moreover, let $\mathcal{S}_0 = \{\mathsf{sig}(f_0)\}$.

Define a digraph $\mathcal{G}=(\bigcup_{i=0}^{k} \mathcal{S}_i, \bigcup_{i=0}^{k-1} E_i)$, where $E_i = \{((x,\sigma), (x', \sigma'))\mid (x, \sigma)\in \mathcal{S}_i, (x', \sigma')\in \mathcal{S}_{i+1}\}$ for $i=0,1,\dots, k-1$.
We suppose that $e=((x,\sigma), (x', \sigma'))$ in $E_i$ has a length equal to the shortest length of a swap sequence from $(x,\sigma)$ to $(x',\sigma')$, which can be computed in polynomial time by Lemma~\ref{lem:shorestpathFor2placements}.

We see that the shortest path from the vertex in $\mathcal{S}_0$ to some vertex in $\mathcal{S}_k$ corresponds to a shortest feasible swap sequence from $f_0$.
The number of vertices of the graph is bounded by $O(k |\mathcal{S}|)$, which is polynomial when $k$ is a constant.
Thus, the theorem holds.
\end{proof}

\subsection{Fixed-parameter algorithm}\label{sec:algorithm_fpt}

In this section, we present a fixed-parameter algorithm parameterized by $k$ by dynamic programming.

We first observe that we can compute the shortest length of a swap sequence between two signatures with the same bijection $\sigma$.

\begin{lemma}\label{lem:pathSamePermutation}
Suppose that we are given two signatures $(x,\sigma)$ and $(y,\sigma)$ with the same bijection $\sigma$.
Then, the shortest length of a swap sequence from $(x,\sigma)$ to $(y,\sigma)$ is equal to
$\sum_{i=1}^{\tilde{k}} |v_i - w_i|$,
where $v_i = \sum_{j=0}^{i-1} x_j$ and $w_i = \sum_{j=0}^{i-1} y_j$ for $i=1,\dots, \tilde{k}$.
Moreover, there exists a shortest swap sequence such that all the token placements in the sequence have the same bijection $\sigma$ in their signatures.
\end{lemma}
\begin{proof}
Since the initial and target token placements have the same $\sigma$ in their signature, we need $|w_i - v_i|$ swaps to move the token $\sigma (i)$ to $w_i$ for any $i=1,\dots, \tilde{k}$. Thus the shortest swap sequence has length at least $\sum_{i=1}^{\tilde{k}} |v_i - w_i|$.
To see that they are equal, we show the existence of a swap sequence of length $\sum_{i=1}^{\tilde{k}} |v_i - w_i|$ by induction on this value.
If $\sum_{i=1}^{\tilde{k}} |v_i - w_i| = 0$, then the claim is obvious. 
Otherwise, let $p$ be the minimum index such that $v_{p} \neq w_{p}$.
By changing the roles of $x$ and $y$ if necessary, we may assume that $v_{p} > w_{p}$. 
Then, starting from $(x,\sigma)$, we apply swap operations $v_p - w_p$ times to obtain a new signature $(x',\sigma)$ 
such that $x'_{p-1} = x_{p-1} - (v_p - w_p)$ and $x'_{p} = x_{p} + (v_p - w_p)$. 
That is, $v'_i = v_i$ for $i \in [\tilde{k}] \setminus \{p\}$ and $v'_p = w_p$, where $v'_i$ is defined as $v'_i = \sum_{j=0}^{i-1} x'_j$. 
Note that this operation is possible without changing the bijection $\sigma$, because $v'_p = w_p > w_{p-1} = v_{p-1}$ by the minimality of $p$. 
By the induction hypothesis, there exists a swap sequence of length $\sum_{i=1}^{\tilde{k}} |v'_i - w_i|$
between $(x',\sigma)$ and $(y,\sigma)$.  
Therefore, we obtain a swap sequence between $(x,\sigma)$ and $(y,\sigma)$
whose length is $|v_p - w_p| + \sum_{i=1}^{\tilde{k}} |v'_i - w_i| = \sum_{i=1}^{\tilde{k}} |v_i - w_i|$. 
Moreover, each token placement in the obtained swap sequence has the same bijection $\sigma$.
\end{proof}

By the lemma, the shortest length of a swap sequence from $(x, \sigma)$ to $(y, \sigma)$ does not depend on $\sigma$.
Thus, we denote it by $d(x, y)$ for $x, y\in\mathcal{R}$.

Let $\mathbf{f}$ be a feasible swap sequence $f_0 \leadsto f_1 \leadsto \dotsm \leadsto f_\ell$.
Suppose that $\varphi(s_j)$ is realized at token placement $f_{i_j}$, that is, $f_{i_j}^{-1}(\varphi(s_j)) \in E$ for $j=1,2,\dots, k$ and $i_1\leq i_2\leq \dots \leq i_k$.
We define $i_0=0$.
Note that the number of distinct values in $\{i_1, \dots, i_k\}$, denoted by $\alpha$, is at most $k$.
Also, let $\beta$ be the number of times that $\sigma$ in the signature changes in the swap sequence.
We call $\alpha + \beta$ the \textit{signature length} of $\mathbf{f}$.

The following lemma says that the signature length is bounded by a function of $k$, which we denote by $\ell_{\max}$.

\begin{lemma}
For a shortest feasible swap sequence $f_0 \leadsto f_1 \leadsto \dotsm \leadsto f_\ell$,
the signature length is bounded by $((2k)!+1)k$ from above.
\end{lemma}
\begin{proof}
Consider the partial swap sequence $f_{i_{j-1}} \leadsto \dotsm \leadsto f_{i_j}$ for $j=1,2,\dots, k$.
We observe that, in the partial swap sequence, the number of times that bijections change is at most $\tilde{k}!$.
Indeed, in this partial sequence, token placements with the same bijection appear sequentially, as otherwise, we can short-cut between them by Lemma~\ref{lem:pathSamePermutation}.
Therefore, the total signature length is bounded by $k(\tilde{k}!)+k \leq k((2k)!)+k$.
\end{proof}

Let $g^{\ell}(x, \sigma, P)$ be the shortest length of a swap sequence to realize $P$ from some token placement $f_0$ with $\mathsf{sig}(f_0)=(x, \sigma)$ such that it has signature length at most $\ell$.
In what follows, we derive a recursive equation on $g^{\ell}(x, \sigma, P)$ for dynamic programming.

We first give notation.
For a bijection $\sigma \in \Sigma$ and a non-negative integer $j$ with $1\leq j\leq \tilde{k}-1$, let $\sigma_j \in \Sigma$ be the bijection obtained from $\sigma$ by swapping the $j$-th and $(j+1)$-st tokens.
We define $\mathcal{R}_j=\{x\in \mathcal{R}\mid x_j = 1\}$, which is the set of signatures such that the $j$-th token and the $(j+1)$-st token are adjacent.

To derive a recursive equation on $g^{\ell}(x, \sigma, P)$, consider the following two cases in which the signature length is decreased at least by one, separately.

The first case is when the bijection $\sigma$ is changed, i.e., $\beta$ decreases.
Suppose that we change $\sigma$ to $\sigma_j$ for some $1\leq j\leq \tilde{k}-1$.
In this case, we first move to $(x', \sigma)$ for some $x'\in\mathcal{R}_j$, and then change $(x', \sigma)$ to $(x', \sigma_j)$. 
By Lemma~\ref{lem:pathSamePermutation}, the number of swaps is $d(x, x')+1$.
After moving to $(x', \sigma_j)$, we can recursively consider finding a shortest swap sequence to realize $P$ from $(x', \sigma_j)$ with the signature length at most $\ell-1$.
Therefore, the total length in this case is $g^{\ell-1} (x', \sigma_j, P) + d (x, x')+1$, and hence the shortest length when we change $\sigma$ is equal to 
\[
\min_{1\leq j\leq \tilde{k}-1} \min_{x'\in \mathcal{R}_j} \left\{ g^{\ell-1} (x', \sigma_j, P) + d (x, x')+1\right\}.
\]

The other case is when $s_1$ is realized without changing $\sigma$, i.e., $\alpha$ decreases.
Then, it is necessary that $s_1=(\sigma(h), \sigma(h+1))$ for some $1\leq h\leq \tilde{k}-1$.
To realize $s_1$, we move $(x, \sigma)$ to $(x', \sigma)$ for some $x'\in \mathcal{R}_h$.
By recursion, the total length in this case is $g^{\ell-1} (x', \sigma, P') + d (x, x')$ by Lemma~\ref{lem:pathSamePermutation}, where $P'$ is the poset obtained from $P$ by removing the first element $s_1$, that is, $P'$ forms the chain $s_2\prec s_3 \prec \dots \prec s_k$.
Thus, the shortest length in this case is 
\[
\min_{x'\in \mathcal{R}_h} \left\{ g^{\ell-1} (x', \sigma, P') + d (x, x')\right\}.
\]

In summary, we have that, for any $x\in \mathcal{R}$, $\sigma \in \Sigma$, and $1\leq \ell \leq \ell_{\max}$,  
\begin{align}\label{eq:rec1}
g^{\ell} (x, \sigma, P) =
\min \biggl\{&
\min_{1\leq j\leq \tilde{k}-1} \min_{x'\in \mathcal{R}_j} \left\{ g^{\ell-1} (x', \sigma_j, P) + d (x, x')+1\right\},\nonumber \\
&\min_{x'\in \mathcal{R}_h} \left\{ g^{\ell-1} (x', \sigma, P') + d (x, x')\right\}
\biggr\},
\end{align}
where $s_1 = (\sigma(h), \sigma(h+1))$ for some $h$. If such $h$ does not exist, the second term is defined to be $+\infty$.

It follows from \eqref{eq:rec1} that we can design a dynamic programming algorithm.
However, the running time would become $O(k\cdot k! \ell_{\max} |\mathcal{R}|)$, and this does not give a fixed-parameter algorithm since 
$|\mathcal{R}|=O(n^{\tilde{k}})$.
In what follows, we will reduce the running time by showing that the minimum is achieved at an extreme point.

For $i = 0, 1, \dots \tilde k -1$, let $e_i$ denote the unit vector whose $i$-th entry is one and the other entries are zeros.
For a vector $x\in \mathcal{R}$ and $1\leq j \leq \tilde{k}-1$, define $N_j (x)$ as
\[
N_j(x) = \{ x + a e_{j-1} - (x_j-1) e_j + b e_{j+1} \mid a + b = x_j-1,\ a, b \in \mathbb{Z}_+\},  
\]
where we regard $e_{\tilde k}$ as the zero vector to simplify the notation. 
Then, $x' \in N_j (x)$ satisfies that
\begin{align*}
x'_j &= 1, \\
x'_{j-1}+x'_{j+1}&=x_{j-1}+x_j+x_{j+1}-1,\\
x'_i & = x_i \ \text{for} \ \ i\not\in\{j-1, j, j+1\},
\end{align*}
where $x_{\tilde{k}}= n-\sum_{i=0}^{\tilde k -1} x_i$ and $x'_{\tilde{k}}= n-\sum_{i=0}^{\tilde k -1} x'_i$.
The signature $(x', \sigma)$ with $x'\in N_j (x)$ means that it is obtained from $(x, \sigma)$ by only moving two tokens $\sigma (j)$ and $\sigma (j+1)$ so that the two tokens are adjacent.
Moreover, define $y^j$ and ${y'}^j$ to be vectors in $N_j(x)$ in which $(a, b) = (0, x_j-1)$ and $(a, b) = (x_j-1, 0)$, respectively. 
Then, 
\begin{align*}
(y^j)_{j-1} = x_{j-1}, &\quad (y^j)_{j+1} = x_j + x_{j+1}-1,\\
({y'}^j)_{j-1} = x_{j-1}+x_j-1, &\quad ({y'}^j)_{j+1} = x_{j+1}.
\end{align*}
Thus, $(y^j, \sigma)$~($({y'}^j, \sigma)$, resp.,) is obtained from $(x, \sigma)$ by only moving one token $\sigma (j+1)$~($\sigma (j)$, resp.,) so that $\sigma (j)$ and $\sigma (j+1)$ are adjacent.

The following theorem asserts that the minimum of~\eqref{eq:rec1} is achieved at either $y^j$ or ${y'}^j$.
\begin{theorem}\label{thm:fptrec}
For any $x\in \mathcal{R}$, $\sigma \in \Sigma$, and $1\leq \ell \leq \ell_{\max}$, it holds that
\begin{align*}
g^{\ell} (x, \sigma, P) =
\min \biggl\{&
\min_{1\leq j\leq \tilde{k}-1, y\in \{y^j, {y'}^j\}} \left\{ g^{\ell-1} (y, \sigma_j, P) + d (x, y)+1\right\},\\
&\min_{y\in \{y^h, {y'}^h\}}\left\{ g^{\ell-1} (y, \sigma, P')+ d (x, y)\right\}
\biggr\}.
\end{align*}
\end{theorem}

To prove the theorem, we show the following two lemmas.
We first show that the minimum of~\eqref{eq:rec1} is achieved at some point in $N_j(x)$.

\begin{lemma}\label{lem:rec2}
It holds that, for any $x\in \mathcal{R}$, $\sigma \in \Sigma$, $1\leq \ell \leq \ell_{\max}$ and $1\leq j\leq \tilde{k}-1$,
\begin{align*}
\min_{x'\in \mathcal{R}_j} \left\{ g^{\ell-1} (x', \sigma, P) + d (x, x')\right\}
=
\min_{y\in N_j(x)} \left\{ g^{\ell-1} (y, \sigma, P) + d(x, y) \right\}.
\end{align*}
\end{lemma}

\begin{proof}
Let $x^\ast$ be a vector in $\mathcal{R}_j$ that attains the minimum of the LHS, that is,
\[
g^{\ell-1} (x^\ast, \sigma, P) + d (x, x^\ast)
=
\min_{x'\in \mathcal{R}_j} \left\{ g^{\ell-1} (x', \sigma, P) + d (x, x')\right\}.
\]
To simplify the notation, we denote $g^\ast = g^{\ell-1} (x^\ast, \sigma, P)$ and $g_y = g^{\ell-1} (y, \sigma, P)$ for a vector $y \in N_j(x)$.

Since $\mathcal{R}_j\supseteq N_j (x)$, it holds that, for any $y \in N_j (x)$,
\begin{align}\label{eq:fpt1}
g^\ast + d (x, x^\ast) \leq g_y + d (x, y).
\end{align}
The minimality of $g_y$ implies that $g_y \leq g^\ast + d(y, x^\ast)$.
Hence, it holds by~\eqref{eq:fpt1} that, for any $y \in N_j (x)$,  
\begin{equation}\label{eq:fpt2}
g^\ast + d (x, x^\ast) \leq g_y + d (x, y) \leq g^\ast + d(x, y) + d(y, x^\ast).
\end{equation}
To show the lemma, it suffices to find $y \in N_j (x)$ that satisfies~\eqref{eq:fpt2} with equality.

By definition, $d(x, y)=x_j -1$ for any $y \in N_j (x)$.
We denote $v_i =\sum_{h=0}^{i-1} x_h$, $v^\ast_i =\sum_{h=0}^{i-1} x^\ast_h$, and $w_i  =\sum_{h=0}^{i-1} y_h$
for any $i=1,\dots,\tilde{k}$.
Then, Lemma~\ref{lem:pathSamePermutation} implies that 
\[
d(x, x^\ast) = \sum_{i=1}^{\tilde{k}} |v_i - v^\ast_i| = X + |v_{j}-v^\ast_j| + |v_{j+1}-v^\ast_{j+1}| 
+ X',
\]
where we define  
$X = \sum_{i=1}^{j-1} |v_i - v^\ast_i|$
and 
$X' = \sum_{i=j+2}^{\tilde{k}} |v_i - v^\ast_i|$.
We distinguish three cases by the position of $v^\ast_j$.
Note that $v^\ast_{j+1} = v^\ast_j + 1$ since $x^\ast \in \mathcal{R}_j$.

First, suppose that $v^\ast_j$ satisfies that $v_j < v^\ast_j< v_{j+1}$.
Then, since $|v_{j}-v^\ast_j| + |v_{j+1}-v^\ast_{j+1}| = v^\ast_{j}-v_j + v_{j+1} - (v^\ast_{j}+1) = x_j -1$, it holds that $d(x, x^\ast) = X + x_j -1 + X'$.
We define $y\in N_j (x)$ so that $w_{j} = v^\ast_j$ and $w_{j+1} = v^\ast_{j+1}$.
Then, we have $d(y, x^\ast) = X + X'$, since $w_i = v_i$ for any $i\in \{1,\dots, j-1\}\cup\{j+2,\dots, \tilde{k}\}$.
Therefore, since $d(x, y) = x_j-1$, we obtain $d(x, x^\ast)=d(y, x^\ast)+d(x, y)$.
Thus, \eqref{eq:fpt2} holds with equality.

Next, suppose that $v^\ast_j$ satisfies that $v^\ast_j\leq  v_j$.
Then it holds that 
\begin{align*}
d(x, x^\ast) &= X +X'+ |v_{j}-v^\ast_j| + |v_{j+1}-v^\ast_{j+1}|\\
&= X + X' + (v_j-v^\ast_j) + (v_{j+1}-v_j + v_j - (v^\ast_{j}+1)) \\
& = X + X' + 2(v_j-v^\ast_j) + x_j-1
\end{align*}
since $v_{j+1}-v_j=x_j$.
We define $y\in N_j (x)$ so that $w_{j} = v_j$ and $w_{j+1} = v_{j}+1$.
Then, we have $d(y, x^\ast) = X + X' + 2(v_j-v^\ast_j)$, since $w_i = v_i$ for any $i\in \{1,\dots, j-1\}\cup\{j+2,\dots, \tilde{k}\}$.
Therefore, since $d(x, y)=x_j -1 $, we obtain $d(x, x^\ast)=d(y, x^\ast)+d(x, y)$, and hence~\eqref{eq:fpt2} holds with equality.

The case where $v^\ast_j\geq  v_{j+1}$ is analogous to the second case.
Therefore, in any case, there exists $y \in N_j (x)$ that satisfies \eqref{eq:fpt2} with equality.
Thus, the lemma holds.
\end{proof}

By Lemma~\ref{lem:rec2}, it holds that
\begin{align}
g^{\ell} (x, \sigma, P) &=
\min \biggl\{
\min_{1\leq j\leq \tilde{k}-1} \min_{y\in N_j(x)} \left\{ g^{\ell-1} (y, \sigma_j, P) + x_j \right\},
\min_{y\in N_h (x)}\left\{ g^{\ell-1} (y, \sigma, P')+ x_h -1 \right\}
\biggr\},\label{eq:fpt3}
\end{align}
in which we note $d(x, y) = x_j-1$ for $y \in N_j(x)$.

\begin{lemma}\label{lem:fptLinearExpression}
The function $g^\ell (x, \sigma, P)$ can be expressed as the minimum of linear functions on $x$.
That is, for any $\sigma \in \Sigma$ and $1\leq \ell \leq \ell_{\max}$, there exist vectors $c_1, \dots, c_p$ and real numbers $\delta_1, \dots, \delta_p$ such that  
\[
g^\ell (x, \sigma, P) = \min_{1\leq i\leq p} \{ c_i x + \delta_i\}
\]
for any $x\in \mathcal{R}$.
\end{lemma}

\begin{proof}
We prove this lemma by induction on $\ell$.
In the base case where $\ell =0$, the poset $P$ is empty, and $\sigma$ is not changed.
Thus, $g^0 (x, \sigma, P) =0$.
Suppose that $\ell \geq 1$.

By~\eqref{eq:fpt3}, it suffices to show that, for any $j$ and $\sigma'$, 
\[
\min_{y\in N_j(x)} \left\{ g^{\ell-1} (y, \sigma', P) + x_j \right\}
\]
can be expressed as the minimum of linear functions.
By the induction hypothesis, $g^{\ell-1}$ is the minimum of linear functions.
Hence, the above can be expressed as
\[
\min_{y\in N_j(x)} \left\{\min_{1\leq i\leq p'} \{ c'_i y + \delta'_i\} + x_j \right\}
=
\min_{1\leq i\leq p'} \min_{y\in N_j(x)} \left\{ c'_i y + x_j + \delta'_i \right\}
\]
for some linear functions $c'_i y +\delta'_i$ for $i=1,\dots, p'$.

Recall that, for every $y \in N_j (x)$, we can write $y = x + a e_{j-1} - (x_j-1) e_j + b e_{j+1}$ for some $a, b \in \mathbb{Z}_+$ with $a + b = x_j -1$. 
Hence, it holds that
\[
c'_i y + x_j + \delta'_i =  
c'_i (x + a e_{j-1} - (x_j-1) e_j + b e_{j+1}) + x_j + \delta'_i,
\]
which is a linear function on $x$ for given $a$, $b$, and $j$.
Therefore, the lemma holds by induction.
\end{proof}

With Lemma \ref{lem:fptLinearExpression}, we are ready to prove Theorem \ref{thm:fptrec}.

\begin{proof}[Proof of Theorem~\ref{thm:fptrec}]
By~\eqref{eq:fpt3}, it suffices to show that, for any $\sigma \in \Sigma$ and $1 \le j \le \tilde{k}-1$, either $y^j$ or ${y'}^j$ achieves the minimum of 
\[
\min_{y\in N_j(x)} \left\{ g^{\ell-1} (y, \sigma, P) + x_j \right\}.
\]
Since $g^{\ell-1}$ is in the form of the minimum of linear functions by Lemma~\ref{lem:fptLinearExpression}, the above can be expressed as
\[
\min_{1\leq i\leq p'} \min_{a, b \in \mathbb{Z}_+: a + b = x_j -1} \left\{ c'_i (x + a e_{j-1} - (x_j-1) e_j + b e_{j+1}) + x_j + \delta'_i \right\}
\]
for some linear functions $c'_i y +\delta'_i$ for $i=1,\dots, p'$.
For given $x$ and $j$, the minimum is attained either at $(a, b)=(x_j-1, 0)$ or at $(a, b)=(0, x_j-1)$.
They correspond to $y^j$ and ${y'}^j$, respectively.
Thus, the theorem holds.
\end{proof}

Let $f_0$ be the initial placement and $\mathsf{sig}(f_0)=(x^0, \sigma^0)$.
Let $v^0_i = \sum_{j=0}^{i-1} x^0_{j}$ for each $i$. 
For $i=0, 1, \dots , \tilde k -1$, 
define $B_i = \{ f_0(v^0_i +1), f_0(v^0_i + 2), \dots , f_0(v^0_{i+1}-1)\}$, which are tokens not in $\tilde{T}$. 
Then, $|B_i| = x^0_i - 1$. 
Theorem~\ref{thm:fptrec} shows that, when we compute $g^{\ell} (x, \sigma, P)$ by using the formula, 
it suffices to consider signatures $(y, \sigma')$ such that 
all tokens in $B_i$ appear consecutively along the path. 
That is, we only consider vectors $y$ such that, for each $i$,  
$y_i - 1 = \sum_{h=i'}^{j'} |B_h|$ for some $0 \le i', j' \le \tilde k -1$ (possibly, $y_i - 1 =0$). 
This shows that each $y_i$ can take one of $O({\tilde k}^2)$ values, and hence 
$y$ has $O({\tilde k}^{2\tilde{k}})$ choices. 
Therefore, 
the size of the DP table is $O(k\cdot k! \ell_{\max} {\tilde k}^{2\tilde{k}})$, and each step can be processed in a fixed-parameter time.
Thus, we have the following theorem.
\begin{theorem}
\label{thm:fpt}
There exists a fixed-parameter algorithm for \textsc{Qubit Routing} when $G$ is a path and $k=|S|$ is a parameter.
\qed
\end{theorem}


\section{Polynomial-time Algorithm for Disjoint Two-Qubit Operations}

We say that an instance $(G, P = (S, \preceq), T, \varphi, f_0)$ of \textsc{Qubit Routing} has \emph{disjoint pairs} 
if $\varphi(s) \cap \varphi(s') = \emptyset$ for any pair of distinct elements $s, s' \in S$. 
The objective of this section is to give a polynomial-time algorithm for instances with disjoint pairs when the graph $G$ is a path. 

\begin{theorem}\label{thm:disjoint}
 \textsc{Qubit Routing} can be solved in polynomial time when a given graph is a path and the instance has disjoint pairs.
\end{theorem}

Let $(G, P = (S, \preceq), T, \varphi, f_0)$ be an instance with disjoint pairs. 
Since $G$ is a path, we suppose for simplicity that $V=\{1,2,\dots, n\}$ and $E=\{\{i, i+1\}\mid i=1,2, \dots, n-1\}$. 
Let $f$ be a token placement.
For an element $s \in S$ with $f^{-1}(\varphi(s)) = \{\alpha_1, \alpha_2\}$, $\alpha_1 < \alpha_2$, we define $\mathsf{gap}_{f}(s) := \alpha_2- \alpha_1 - 1$. 
Then $\mathsf{gap}_{f}(s)$ is a lower bound on the number of swaps to realize $\varphi(s)$ if the initial token placement is $f$.
For distinct elements $s, s' \in S$ such that $f^{-1}(\varphi(s)) = \{\alpha_1, \alpha_2\}$, $\alpha_1 < \alpha_2$, and $f^{-1}(\varphi(s')) = \{\beta_1, \beta_2\}$, $\beta_1 < \beta_2$, 
we say that $s$ and $s'$ \emph{cross} if $\alpha_1 < \beta_1 < \alpha_2 < \beta_2$ or $\beta_1 < \alpha_1 < \beta_2 < \alpha_2$. 
One can observe that, if $S=\{s, s'\}$ such that $s$ and $s'$ cross, we can realize both $\varphi(s)$ and $\varphi(s')$ by $\mathsf{gap}_{f}(s)+\mathsf{gap}_{f}(s')-1$ swaps.
Note that $\alpha_1, \alpha_2, \beta_1$, and $\beta_2$ are always distinct since the instance has disjoint pairs. 
We define 
\begin{align*}
\mathsf{gap}(f) &:= \sum_{s \in S} \mathsf{gap}_{f}(s), \\
\mathsf{cross}(f) &:= \left|\left\{ \{s, s'\} \in \binom{S}{2} \Bigm\vert \mbox{$s$ and $s'$ cross} \right\}\right|, \\
\mathsf{value}(f) &:= \mathsf{gap}(f) - \mathsf{cross}(f), 
\end{align*}
where $\mathsf{gap}$, $\mathsf{cross}$, and $\mathsf{value}$ are regarded as functions only in $f$ by fixing $G$, $P$, $T$, and $\varphi$.

We prove the following proposition, which completes the proof of Theorem \ref{thm:disjoint}.  
\begin{proposition}\label{prop:disjointvalue}
The optimal value for the instance $(G, P = (S, \preceq), T,  \varphi, f_0)$ is $\mathsf{value}(f_0)$.
Such a sequence can be constructed in polynomial time.
\end{proposition}

To prove Proposition \ref{prop:disjointvalue},
we first show that the optimal value is at most $\mathsf{value}(f_0)$. 

\begin{lemma}\label{lem:disjointupper}
The instance $(G, P = (S, \preceq), T, \varphi, f_0)$ has a feasible swap sequence of length $\mathsf{value}(f_0)$, and such a swap sequence can be found in polynomial time.
\end{lemma}

\begin{proof}
We prove the statement by induction on the lexicographic ordering of $(\mathsf{gap}(f_0), \mathsf{value}(f_0))$. 
Note that $\mathsf{gap}(f_0) \ge 0$ and $\mathsf{value}(f_0) \ge - |\binom{S}{2}|$.\footnote{Although we can show a better lower bound on $\mathsf{value}(f_0)$ by careful analysis, 
we just present a trivial lower bound here, because we only need the fact that $\mathsf{value}(f_0)$ has a finite lower bound to apply the induction.}

If $\mathsf{gap}(f_0) = 0$, then $f_0^{-1}(\varphi(x))$ forms a pair of adjacent vertices for any $x \in S$, and hence $\mathsf{cross}(f_0) = \mathsf{value}(f_0)=0$ holds. 
In such a case, since the trivial swap sequence consisting of a single token placement $f_0$ is a feasible sequence of length zero, the statement holds. 

Suppose that $\mathsf{gap}(f_0) \ge 1$. 
Let 
\[
i^* := \max \{i \in \{1, 2, \dots , n\} \mid  \mbox{$\exists s^* \in S$ s.t.~$f_0^{-1}(\varphi(s^*)) = \{i, j\}$ with $j \ge i+2$} \},    
\]
which is well-defined as $\mathsf{gap}(f_0) \ge 1$. 
Let $f_1$ be the token placement obtained from $f_0$ by applying a swap operation on $\{i^*, i^*+1\} \in E$. 
Then, we see that $\mathsf{value}(f_1) = \mathsf{value}(f_0) -1$ and $\mathsf{gap}(f_1) \le \mathsf{gap}(f_0)$ by the following case analysis. 
\begin{enumerate}
\item
Suppose that $i^*+1 \not\in f_0^{-1}(\varphi(s))$ for any $s \in S$. In this case, $\mathsf{gap}(f_1) = \mathsf{gap}(f_0) -1$ and $\mathsf{cross}(f_1) = \mathsf{cross}(f_0)$, 
which shows that $\mathsf{value}(f_1) = \mathsf{value}(f_0) -1$. 
\item
Suppose that $i^*+1 \in f_0^{-1}(\varphi(s))$ for some $s \in S$. By the maximality of $i^*$ and by the assumption that the instance has disjoint pairs, 
we obtain $f_0^{-1}(\varphi(s)) = \{j, i^*+1\}$ for $j < i^*$ or $j = i^*+2$. 
\begin{itemize}
\item
If $f_0^{-1}(\varphi(s)) = \{j , i^*+1\}$ for some $j < i^*$, then $\mathsf{gap}(f_1) = \mathsf{gap}(f_0) -2$ and $\mathsf{cross}(f_1) = \mathsf{cross}(f_0)-1$, 
which shows that $\mathsf{value}(f_1) = \mathsf{value}(f_0) -1$. 
\item
If $f_0^{-1}(\varphi(s)) = \{i^*+1, i^*+2\}$, then $\mathsf{gap}(f_1) = \mathsf{gap}(f_0)$ and $\mathsf{cross}(f_1) = \mathsf{cross}(f_0)+1$, 
which shows that $\mathsf{value}(f_1) = \mathsf{value}(f_0) -1$. 
\end{itemize}
\end{enumerate}

Since $(\mathsf{gap}(f_1), \mathsf{value}(f_1))$ is lexicographically smaller than $(\mathsf{gap}(f_0), \mathsf{value}(f_0))$, 
by the induction hypothesis,
instance $(G, P = (S, \preceq), T, \varphi, f_1)$ has a feasible swap sequence of length $\mathsf{value}(f_1) = \mathsf{value}(f_0) -1$. 
Since $f_1$ is obtained from $f_0$ by a single swap operation, this shows that 
$(G, P = (S, \preceq), T, \varphi, f_0)$ has a feasible swap sequence of length $\mathsf{value}(f_0)$. 
\end{proof}

The proof of Lemma \ref{lem:disjointupper} shows that, by applying $\mathsf{value}(f_0)$ swap operations, 
we can transform $f_0$ into another token placement $f$
such that $f^{-1}(\varphi(s))$ forms a pair of adjacent vertices for any $s \in S$. 
This means that the obtained sequence is a feasible solution for any poset $P$ on $S$. 

We next show that the optimal value is at least $\mathsf{value}(f_0)$. 

\begin{lemma}\label{lem:disjointlower}
Each feasible swap sequence for the instance $(G, P = (S, \preceq), T, \varphi, f_0)$ has length at least $\mathsf{value}(f_0)$. 
\end{lemma}

\begin{proof}
It suffices to show that no single swap operation can decrease $\mathsf{value}(f)$ by more than one. 
Consider a single swap operation on $\{i, i+1\} \in E$ that transforms $f$ into $f'$. 
It is easy to see that $\mathsf{gap}(f') \ge \mathsf{gap}(f) -2$ and $\mathsf{cross}(f') \le \mathsf{cross}(f)+1$. 
We now prove $\mathsf{value}(f') \ge \mathsf{value}(f) -1$ by the case analysis. 
\begin{enumerate}
\item
Suppose that $\mathsf{gap}(f') \ge \mathsf{gap}(f)$. 
In this case, since $\mathsf{cross}(f') \le \mathsf{cross}(f)+1$, we obtain $\mathsf{value}(f') \ge \mathsf{value}(f) -1$. 

\item
Suppose that $\mathsf{gap}(f') = \mathsf{gap}(f)-1$. 
In this case, we have one of the following.  
\begin{itemize}
\item
There exists $s \in S$ such that $f^{-1}(\varphi(s)) = \{j, i+1\}$ for some $j \le i-1$, and $i \not\in f^{-1}(\varphi(s'))$ for any $s' \in S$, or
\item
There exists $s \in S$ such that $f^{-1}(\varphi(s)) = \{i, j\}$ for some $j \ge i+2$, and $i+1 \not\in f^{-1}(\varphi(s'))$ for any $s' \in S$. 
\end{itemize}
In both cases, we obtain $\mathsf{cross}(f') = \mathsf{cross}(f)$, which shows that $\mathsf{value}(f') = \mathsf{value}(f) -1$. 

\item
Suppose that $\mathsf{gap}(f') = \mathsf{gap}(f)-2$. 
In this case, 
there exist $s, s' \in S$ such that $f^{-1}(\varphi(s)) = \{j, i+1\}$ for some $j \le i-1$ and $f^{-1}(\varphi(s')) = \{i, k\}$ for some $k \ge i+2$. 
Then, we obtain $\mathsf{cross}(f') = \mathsf{cross}(f)-1$, which shows that $\mathsf{value}(f') = \mathsf{value}(f) -1$. 
\end{enumerate}

This completes the proof of the lemma. 
\end{proof}

\section{Hardness: Stars and Antichains}\label{sec:stars_antichains}

In this section, we show that the problem is NP-hard even when $G$ is a star and $P$ is an antichain.
Recall that an antichain is a poset in which every pair of elements is incomparable.

\begin{theorem}
\label{thm:hard-star}
 \textsc{Qubit Routing} is NP-hard even when $G$ is a star and $P$ is an antichain.
\end{theorem}
\begin{proof}
    We reduce \textsc{Vertex Cover}, which is known to be NP-hard~\cite{GAREY1976237}.

    \begin{mdframed}
    \noindent
    \textsc{Vertex Cover}
    \begin{description}
    \item[Input.] A graph $H=(V(H), E(H))$ and a positive integer $k$. 
    \item[Question.] Is there a vertex subset $X \subseteq V(H)$, called a \emph{vertex cover} of $H$, such that $|X|\leq k$ and $\{u,v\}\cap X \neq \emptyset$ for every edge $\{u,v\} \in E(H)$?
    \end{description}
    \end{mdframed}

    Suppose that we are given an instance of \textsc{Vertex Cover} that consists of a graph $H=(V(H), E(H))$ and a positive integer $k$.
    Let $n = |V(H)|$. 
    We construct an instance of \textsc{Qubit Routing} as follows.
    Define a set of tokens as $T=V(H) \cup \{0\}$, where $0\not\in V(H)$.
    Let $S = E(H)$ and define $\varphi\colon S \to \binom{T}{2}$ as $\varphi(\{u,v\}) = \{u,v\}$ for each $\{u,v\} \in S$.
    The poset $P$ is an antichain, i.e., every pair of elements in $S$ is incomparable.
    Define a graph $G=(V, E)$ as $V = \{0,1,\dots,n\}$ and $E = \{\{0,i\} \mid i \in \{1,2,\dots,n\}\}$.
    The initial token placement $f_0\colon V\to T$ should satisfy $f_0(0) = 0$ and other tokens can be placed arbitrarily over the vertices $\{1,2,\dots,n\}$.

    We claim that there exists a swap sequence of length $k$ for the instance $(G, P, T, \varphi, f_0)$ of \textsc{Qubit Routing} if and only if there exists a vertex cover $X$ in $H$ of size $k$.
    Suppose that there exists a vertex cover $X=\{v_1, \dots, v_k\}$ in $H$ of size $k$.
    Then, we construct the following swap sequence.
    In the swap $f_{i-1} \leadsto f_i$, we exchange the token at the vertex $0$ and the token $v_i$.
    Then, for each edge $\{v_i, v\} \in E(H)$, it holds that $\{f^{-1}_i(v_i), f^{-1}_i(v)\} = \{0, f^{-1}_i(v)\} \in E$.
    Therefore, the swap sequence $f_0 \leadsto \dots \leadsto f_k$ realizes $P$.
    Conversely, suppose that there exists a swap sequence $f_0 \leadsto f_1 \leadsto \dots \leadsto f_k$ of length $k$ for the instance $(G, P, T, \varphi, f_0)$.
    Let $X=\{v \in V(H) \mid f_i^{-1}(0) = v \text{ for some } i \in \{1,2,\dots,k\}\}$, i.e., $X$ is the set of vertices of $H$ that are placed at $0$ in the course of swaps.
    Since the swap sequence realizes $P$, for every edge $\{u,v\} \in E(H)$, there exists $f_i$ such that 
    $\{f^{-1}_i(u), f^{-1}_i(v)\} \in E$, which implies that $u \in X$ or $v \in X$.
    Therefore, $X$ is a vertex cover of $H$.    
\end{proof}

\section{Concluding Remarks}

We initiated algorithmic studies on the quantum routing problem, also known as the swap minimization problem, from the viewpoint of theoretical computer science.
The problem is of central importance in compiler design for quantum programs when they are implemented in some of the superconducting quantum computers such as IBM Quantum systems.

Most notably, we proved the quantum routing problem is NP-hard even when the graph topology of a quantum computer is a path, which corresponds to the so-called linear nearest neighbor architecture.
In our proof, the initial token placement can be chosen arbitrarily.
This implies that the combined optimization of the quantum assignment and the quantum routing is also NP-hard for the same architecture.

We also gave some algorithmic results, but they were restricted to the case of the linear nearest neighbor architectures.
Possible future work is to give algorithmic results with theoretical guarantees for other graph topologies.

\paragraph{Acknowledgment}

We thank Toshinari Itoko at IBM Research Tokyo for bringing the qubit allocation problem to our attention and also for his valuable comments.

\bibliographystyle{abbrvurl}
\bibliography{qubit}

\begin{thebibliography}{10}

\bibitem{DBLP:conf/esa/AichholzerDKLLM22}
O.~Aichholzer, E.~D. Demaine, M.~Korman, A.~Lubiw, J.~Lynch,
  Z.~Mas{\'{a}}rov{\'{a}}, M.~Rudoy, V.~V. Williams, and N.~Wein.
\newblock Hardness of token swapping on trees.
\newblock In S.~Chechik, G.~Navarro, E.~Rotenberg, and G.~Herman, editors, {\em
  30th Annual European Symposium on Algorithms, {ESA} 2022, September 5-9,
  2022, Berlin/Potsdam, Germany}, volume 244 of {\em LIPIcs}, pages 3:1--3:15.
  Schloss Dagstuhl - Leibniz-Zentrum f{\"{u}}r Informatik, 2022.
\newblock \href {https://doi.org/10.4230/LIPIcs.ESA.2022.3}
  {\path{doi:10.4230/LIPIcs.ESA.2022.3}}.

\bibitem{DBLP:journals/qip/AsakaSY20}
R.~Asaka, K.~Sakai, and R.~Yahagi.
\newblock Quantum circuit for the fast {F}ourier transform.
\newblock {\em Quantum Inf. Process.}, 19(8):277, 2020.
\newblock \href {https://doi.org/10.1007/s11128-020-02776-5}
  {\path{doi:10.1007/s11128-020-02776-5}}.

\bibitem{DBLP:journals/algorithmica/BonnetMR18}
{\'{E}}.~Bonnet, T.~Miltzow, and P.~Rzazewski.
\newblock Complexity of token swapping and its variants.
\newblock {\em Algorithmica}, 80(9):2656--2682, 2018.
\newblock \href {https://doi.org/10.1007/s00453-017-0387-0}
  {\path{doi:10.1007/s00453-017-0387-0}}.

\bibitem{DBLP:conf/socs/BoteaK018}
A.~Botea, A.~Kishimoto, and R.~Marinescu.
\newblock On the complexity of quantum circuit compilation.
\newblock In V.~Bulitko and S.~Storandt, editors, {\em Proceedings of the
  Eleventh International Symposium on Combinatorial Search, {SOCS} 2018,
  Stockholm, Sweden - 14-15 July 2018}, pages 138--142. {AAAI} Press, 2018.

\bibitem{GAREY1976237}
M.~Garey, D.~Johnson, and L.~Stockmeyer.
\newblock Some simplified {NP}-complete graph problems.
\newblock {\em Theoretical Computer Science}, 1(3):237--267, 1976.
\newblock \href {https://doi.org/10.1016/0304-3975(76)90059-1}
  {\path{doi:10.1016/0304-3975(76)90059-1}}.

\bibitem{DBLP:journals/jgaa/KawaharaSY19}
J.~Kawahara, T.~Saitoh, and R.~Yoshinaka.
\newblock The time complexity of permutation routing via matching, token
  swapping and a variant.
\newblock {\em J. Graph Algorithms Appl.}, 23(1):29--70, 2019.
\newblock \href {https://doi.org/10.7155/jgaa.00483}
  {\path{doi:10.7155/jgaa.00483}}.

\bibitem{DBLP:conf/asplos/LiDX19}
G.~Li, Y.~Ding, and Y.~Xie.
\newblock Tackling the qubit mapping problem for {NISQ}-era quantum devices.
\newblock In I.~Bahar, M.~Herlihy, E.~Witchel, and A.~R. Lebeck, editors, {\em
  Proceedings of the Twenty-Fourth International Conference on Architectural
  Support for Programming Languages and Operating Systems, {ASPLOS} 2019,
  Providence, RI, USA, April 13-17, 2019}, pages 1001--1014. {ACM}, 2019.
\newblock \href {https://doi.org/10.1145/3297858.3304023}
  {\path{doi:10.1145/3297858.3304023}}.

\bibitem{DBLP:conf/esa/MiltzowNORTU16}
T.~Miltzow, L.~Narins, Y.~Okamoto, G.~Rote, A.~Thomas, and T.~Uno.
\newblock Approximation and hardness of token swapping.
\newblock In P.~Sankowski and C.~D. Zaroliagis, editors, {\em 24th Annual
  European Symposium on Algorithms, {ESA} 2016, August 22-24, 2016, Aarhus,
  Denmark}, volume~57 of {\em LIPIcs}, pages 66:1--66:15. Schloss Dagstuhl -
  Leibniz-Zentrum f{\"{u}}r Informatik, 2016.
\newblock \href {https://doi.org/10.4230/LIPIcs.ESA.2016.66}
  {\path{doi:10.4230/LIPIcs.ESA.2016.66}}.

\bibitem{10.1145/3544563}
G.~Nannicini, L.~S. Bishop, O.~G\"{u}nl\"{u}k, and P.~Jurcevic.
\newblock Optimal qubit assignment and routing via integer programming.
\newblock {\em ACM Transactions on Quantum Computing}, 4(1), oct 2022.
\newblock \href {https://doi.org/10.1145/3544563} {\path{doi:10.1145/3544563}}.

\bibitem{DBLP:journals/qip/SaeediWD11}
M.~Saeedi, R.~Wille, and R.~Drechsler.
\newblock Synthesis of quantum circuits for linear nearest neighbor
  architectures.
\newblock {\em Quantum Inf. Process.}, 10(3):355--377, 2011.
\newblock \href {https://doi.org/10.1007/s11128-010-0201-2}
  {\path{doi:10.1007/s11128-010-0201-2}}.

\bibitem{DBLP:conf/cgo/SiraichiSCP18}
M.~Y. Siraichi, V.~F. dos Santos, C.~Collange, and F.~M.~Q. Pereira.
\newblock Qubit allocation.
\newblock In J.~Knoop, M.~Schordan, T.~Johnson, and M.~F.~P. O'Boyle, editors,
  {\em Proceedings of the 2018 International Symposium on Code Generation and
  Optimization, {CGO} 2018, V{\"{o}}sendorf / Vienna, Austria, February 24-28,
  2018}, pages 113--125. {ACM}, 2018.
\newblock \href {https://doi.org/10.1145/3168822} {\path{doi:10.1145/3168822}}.

\bibitem{DBLP:journals/pacmpl/SiraichiSCP19}
M.~Y. Siraichi, V.~F. dos Santos, C.~Collange, and F.~M.~Q. Pereira.
\newblock Qubit allocation as a combination of subgraph isomorphism and token
  swapping.
\newblock {\em Proc. {ACM} Program. Lang.}, 3({OOPSLA}):120:1--120:29, 2019.
\newblock \href {https://doi.org/10.1145/3360546} {\path{doi:10.1145/3360546}}.

\bibitem{DBLP:journals/qip/HouteMACP20}
R.~van Houte, J.~Mulderij, T.~Attema, I.~Chiscop, and F.~Phillipson.
\newblock Mathematical formulation of quantum circuit design problems in
  networks of quantum computers.
\newblock {\em Quantum Inf. Process.}, 19(5):141, 2020.
\newblock \href {https://doi.org/10.1007/s11128-020-02630-8}
  {\path{doi:10.1007/s11128-020-02630-8}}.

\bibitem{DBLP:journals/tcs/YamanakaDIKKOSS15}
K.~Yamanaka, E.~D. Demaine, T.~Ito, J.~Kawahara, M.~Kiyomi, Y.~Okamoto,
  T.~Saitoh, A.~Suzuki, K.~Uchizawa, and T.~Uno.
\newblock Swapping labeled tokens on graphs.
\newblock {\em Theor. Comput. Sci.}, 586:81--94, 2015.
\newblock \href {https://doi.org/10.1016/j.tcs.2015.01.052}
  {\path{doi:10.1016/j.tcs.2015.01.052}}.

\end{thebibliography}

\end{document}